\documentclass{acmsmalltr}
\usepackage{verbatim}
\usepackage{graphicx}
\usepackage{xspace}
\usepackage{multirow}
\usepackage{rotating}
\usepackage{amsmath}
\newtheorem{sanity}{Sanity check}

\def\..{\,\mathpunct{\ldotp\ldotp}} 

\newcommand{\Z}{\mathbf Z}

\newcommand{\lstt}[2]{${#1}_0$,~${#1}_1$, $\dots\,$,~${#1}_{#2-2}$}
\newcommand{\xorshift}[1][]{\texttt{xorshift#1}\xspace}
\newcommand{\xorshifts}[1][]{\texttt{xorshift#1*}\xspace}
\newcommand{\wella}{\texttt{WELL1024a}\xspace}

\newcommand{\wellb}{\texttt{WELL19937a}\xspace}
\newcommand{\mt}[1][]{\texttt{MT19937}\xspace}
\newcommand{\xorgens}[1][]{\texttt{xorgens#1}\xspace}
\newcommand{\xst}[3]{#1, #2, #3}
\newcommand{\xs}[4]{$A_{#1}(#2, #3, #4)$}

\begin{document}
\markboth{S.~Vigna}{An experimental exploration of Marsaglia's \xorshift
generators, scrambled}

\bibliographystyle{ACM-Reference-Format-Journals}

\title{An experimental exploration of Marsaglia's \xorshift generators, scrambled}
\author{Sebastiano Vigna
\affil{Universit\`a degli Studi di Milano, Italy}}

\begin{abstract}
Marsaglia proposed \xorshift generators as a class of very fast,
good-quality pseudorandom number generators. Subsequent analysis by 
Panneton and L'Ecuyer has lowered the expectations raised by Marsaglia's
paper, showing several weaknesses of such generators. Nonetheless, many of the weaknesses
of \xorshift generators fade away if their result is scrambled by a non-linear
operation (as originally suggested by Marsaglia). In this paper we explore the space of possible generators obtained 
by multiplying the result of a \xorshift generator by a suitable constant. We
sample generators at 100 points of their state space and obtain detailed statistics that lead
us to choices of parameters that improve on the current ones. 
We then explore for the first time the space of high-dimensional \xorshift generators, following
another suggestion in Marsaglia's paper, finding choices of parameters providing periods
of length $2^{1024}-1$ and $2^{4096}-1$. The resulting generators are of extremely high
quality, faster than current similar alternatives, and generate long-period sequences passing
strong statistical tests using only eight logical operations, one addition and one multiplication by a constant.
\end{abstract}




\acmformat{Sebastiano Vigna, 2014. An experimental exploration of Marsaglia's \xorshift generators, scrambled.}

\maketitle

\begin{bottomstuff}
This work is supported the EU-FET grant NADINE (GA 288956).

This paper is an extended version of the paper with the same title published
in the ACM Transactions on Mathematical Software~\cite{VigEEMXGS}.

Author's addresses: Sebastiano Vigna, Dipartimento di Informatica,
Universit\`a degli Studi di Milano, via Comelico 39, 20135 Milano MI, Italy.
\end{bottomstuff}

\section{Introduction}

\xorshift generators are a simple class of pseudorandom number generators
introduced by~\citeN{MarXR}. In Marsaglia's view, their main feature is
speed: in particular, a \xorshift generator with a 64-bit state
generates a new 64-bit value using just three 64-bit shifts and three 64-bit xors (i.e., exclusive ors), thus
making it possible to generate hundreds of millions of values per
second.

Subsequent analysis by~\citeN{BreNMXRNG} showed that the bits generated by
\xorshift generators are equivalent to certain \emph{linear feedback shift registers}.
\citeN{PaLXRNG} analyzed in detail the theoretical properties of the generators,
and found empirical weaknesses using the TestU01 suite~\cite{LESTU01}. They
proposed an increase in the number of shifts, or combination with another generator, to improve quality.

In the first part of this paper, as warm-up we explore experimentally 
the space of \xorshift generators with 64 bits of state using statistical test suites. We sample
generators at 100 points of their state space, to easily identify spurious failures.
Marsaglia proposes some choice of parameters, that, as we
will see, and as already reported by~\citeN{PaLXRNG}, are not particularly good.
We report results that are actually worse than those of Panneton and L'Ecuyer as
we use the entire 64-bit output of the generators.
While we can suggest some good parameter choices, the result remains poor.

Thus, we turn to the idea of scrambling the result of a \xorshift generator
using a multiplication, as it is typical, for instance, in the construction of
practical hash functions due to the resulting \emph{avalanching} behavior (bits
of the result depend on several bits of the input). This method 
is actually suggested in passing in Marsaglia's paper.
The third edition of the classic ``Numerical Recipes''~\cite{PTVNR}, indeed,
proposes this construction for a basic, all-purpose generator.
From the wealth of data so obtained we derive generators with 
better statistical properties than those suggested in ``Numerical Recipes''.

In the last part of the paper, we follow the suggestion about high-dimensional
generators contained in Marsaglia's paper, and compute 
several choices of parameters that provide full-period \xorshift generators with a state of $1024$
and $4096$ bits. Once again, we propose generators that use a multiplication
to scramble the result.

At the end of the paper, we apply the same methodology to a number of popular
non-cryptographic generators, and we discover that our high-dimensional
generators are actually faster and of higher or equivalent statistical quality,
as assessed by statistical test suites, than the alternatives.

The software used to perform the experiments described in this paper is
distributed by the author under the GNU General Public License. Moreover,
all files generated during the experiments are available from the author. They contain a large
amount of data that could be further analyzed (e.g., by studying the
distribution of $p$-values over the seeds). We leave this issue open
for further work.

\begin{figure}
\centering
\begin{tabular}{l|c|l}
& C code & \\
\hline
$A_0$ & \verb.x ^= x << a; x ^= x >> b; x ^= x << c;.  & $\mathbf{X}_1$\\
$A_1$ & \verb.x ^= x >> a; x ^= x << b; x ^= x >> c;. & $\mathbf{X}_3$\\
$A_2$ & \verb.x ^= x << c; x ^= x >> b; x ^= x << a;. & $\mathbf{X}_2$\\
$A_3$ & \verb.x ^= x >> c; x ^= x << b; x ^= x >> a;. & $\mathbf{X}_4$\\
$A_4$ & \verb.x ^= x << a; x ^= x << c; x ^= x >> b;. & $\mathbf{X}_5$\\
$A_5$ & \verb.x ^= x >> a; x ^= x >> c; x ^= x << b;. & $\mathbf{X}_6$\\
$A_6$ & \verb.x ^= x >> b; x ^= x << a; x ^= x << c;. & $\mathbf{X}_7$\\
$A_7$ & \verb.x ^= x << b; x ^= x >> a; x ^= x >> c;. & $\mathbf{X}_8$
\end{tabular}
\caption{\label{tab:algo}The eight possible \xorshift[64] algorithms.  
The list is actually derived from~\protect\citeN{PaLXRNG},
as they correctly remarked that two of the eight algorithms proposed 
by Marsaglia were redundant, whereas two ($A_6$ and $A_7$) were missing.
On the right side we report the name of the linear transformation associated to
the algorithm as denoted by~\protect\citeN{PaLXRNG}. With our numbering,
algorithms $A_{2i}$ and $A_{2i+1}$ are conjugate by reversal. Note that contiguous shifts
in the same direction can be exchanged without affecting the resulting
algorithm. We normalized such contiguous shifts so that their letters are
lexicographically sorted.}
\end{figure}

\section{An introduction to \xorshift generators}

The basic idea of \xorshift generators is that their state is modified by
applying repeatedly a shift and an exclusive-or (xor) operation. In this paper we consider
64-bit shifts and states made of $2^n$ bits, with $n\geq 6$. We usually append
$n$ to the name of a family of generators when we need to restrict the discussion 
to a specific state size.

For \xorshift[64] generators Marsaglia suggests a number of possible combination of shifts,
shown in Figure~\ref{tab:algo}. Not all choices of parameters give a full ($2^{64}-1$)
period: there are $275$ suitable choices of $a$, $b$ and $c$ and eight variants,
totaling 2200 generators.
 
In linear-algebra terms, if $L$ is the $64\times 64$ matrix on $\Z/2\Z$ that effects a left shift of
one position on a binary row vector (i.e., $L$ is all zeroes except for ones on
the principal subdiagonal) and if $R$ is the right-shift
matrix (the transpose of $L$), each left/right shift and xor can be described as
a linear multiplication by $\bigl(I+L^s\bigr)$ or $\bigl(I+R^s\bigr)$, respectively, where $s$ is the
amount of shifting.\footnote{A more detailed study of the linear algebra behind \xorshift generators can be found in~\cite{MarXR,PaLXRNG}.} For instance, algorithm $A_0$ of Figure~\ref{tab:algo} is equivalent to the
$\Z/2\Z$-linear transformation
\[
\mathbf{X}_1 = \bigl(I+L^a\bigr)\bigl(I+R^b\bigr)\bigl(I+L^c\bigr). 
\]
It is useful to associate with a linear transformation $M$ its \emph{characteristic polynomial}
\[
P(x)=\operatorname{det}(M-xI).
\]
The associated generator has maximum-length period if and only if $P(x)$ is primitive over $\Z/2\Z$. 
This happens if $P(x)$ is irreducible and if $x$ has
maximum period in the ring of polynomial over $\Z/2\Z$ modulo $P(x)$, that is, if
the powers $x$,~$x^2$, $\dots\,$,~$x^{2^n-1}$ are distinct modulo $P(x)$. Finally,
to check the latter condition is sufficient to check that
\[x^{(2^n-1)/p}\neq 1\mod P(x)\]
for every prime $p$ dividing $2^n-1$~\cite{LiNIFFA}. 

The \emph{weight} of $P(x)$ is the number
of terms in $P(x)$, that is, the number of nonzero coefficients. It is considered a good property for generators
of this kind that the weight is close to $n/2$, that is, that the polynomial
is neither too sparse nor too dense~\cite{ComHCRBS}.

Note that the family of algorithms of Figure~\ref{tab:algo} is intended to
generate \emph{64-bit values}. This means that the entire output of the
algorithm should be used when performing tests. We will see that this has not
always been the case in previous literature.

\section{Setting up the experiments}
\label{sec:setup}

In this paper we want to explore experimentally the space of a number of
\xorshift-based generators. Our purpose is to identify variants
with full period which have particularly good statistical properties, and test
whether claims about good parameters made in the previous literature are
confirmed.

The basic idea is that of \emph{sampling} the generators by executing a battery
of tests starting with 100 different seeds that are equispaced in the state
space. More precisely, if the state is made of $n$ bits we
use the seeds $1+i\lfloor 2^{n}/100\rfloor$, $0\leq i<100$. The tests produce a
number of statistics, and we decided to use as \emph{score} the number of failed
tests. A higher score, thus, means lower quality.
Running multiple tests makes it easy to rule out spurious failures, as suggested
also by~\citeN{RSNSTSRPNGCA} in the context of cryptographic
applications.\footnote{We remark that, arguably, a more principled
choice would be choosing seeds that are equispaced \emph{in the sequence of states
traversed by the generator}. Unfortunately, this is possible only for
generators with ``jump-ahead'' primitives, and we want our methodology to be
universal. We checked that all sequences of states used in our tests
on generators with $64$ bits of state do not overlap. The chance that this
happens with more than $128$ bits of state is negligible.}

We use two tools to perform our tests. The first and most important is
TestU01, a test suite developed by~\citeN{LESTU01} that
contains several tests oriented towards the generation of uniform real numbers
in $[0\..1)$.\footnote{We use the double-dot notation for intervals introduced
by C.\,A.\,R.~Hoare and Lyle Ramshaw~\cite{GKPCM}.}
We also perform tests using Dieharder, a suite of tests developed by \citeN{BroD}, both as a sanity check and to compare the power of the two suites. 
Dieharder contains all original tests from Marsaglia's
Diehard, plus many more additional tests.
We refer frequently to the specific type of tests failed:
the reader can refer to the TestU01 and Dieharder documentation for more
information. 

We consider a test failed if its $p$-value is outside of the interval
$[0.001\..0.999]$. This is the interval outside which TestU01 reports a test by
default.
Sometimes a much stricter threshold is used (For instance,
\citeN{LESTU01} use $[10^{-10}\..1-10^{-10}]$ when applying TestU01 to 
a variety of generators), and weaker $p$-values are called \emph{suspicious
values}, but since we are going to repeat the
test $100$ times we can use relatively weak $p$-values: spurious failures will
appear rarely, and we can catch borderline cases (e.g., tests failing on 50\% of
the seeds) that give us useful information. 

We call \emph{systematic} a failure that happens for all seeds. For all such
failures in our tests, $p$-values are smaller than $10^{-15}$. Thus, all
conclusions drawn in this paper based on systematic failures would not change
even if we lowered significantly the failure threshold. More generally, $90$\%
of the $p$-values of failed tests are actually smaller than $10^{-6}$.

We remark that our choice (counting the number of failures) is somewhat rough;
for example, we consider the same failure a $p$-value very close to $0$ and a
$p$-value just below $0.001$. Indeed, other, more sophisticated methods might be
used to aggregate the result of our samples: combining $p$-values, for instance,
or computing a $p$-value of $p$-values~\cite{RSNSTSRPNGCA}.
However, our choice is very easy to interpret, and multiple samples partially
compensate this problem (spurious failures will appear in few samples).

Of course, the number of experiments is very large---in fact, our experiments
were carried out using hundreds of cores in parallel and, overall, they add up
to more than a century of computational time. Our strategy is to apply a very fast
test to all generators and seeds, in the hope of isolating a small group of
generators that behave significantly better with respect to these tests.
Stronger tests can then be applied to this subset. The same strategy has been followed by~\citeN{PanPhD} in the experimental study
of \xorshift generators contained in his Ph.D.~thesis.

TestU01 offers three different predefined batteries of tests (SmallCrush, Crush
and BigCrush) with increasing computational cost and increased difficulty.
Unfortunately, Dieharder does not provide such a segmentation.

Note that Dieharder has a concept of ``weak'' success and a concept of ``failure'',
depending on the $p$-value of the test, and we used command-line options to align
its behavior with that of TestU01: a $p$-value outside of the range
$[0.001\..0.999]$ is a failure. Moreover, we disabled the initial timing tests so
that exactly the same stream of 64-bit numbers is fed to the two test suites.

In both cases we implemented our own \xorshift generator. Some care is needed in
this phase, as both TestU01 and Dieharder are inherently 32-bit test suites:
since we want to test \xorshift as a \emph{64-bit} generator, it is important
that all bits produced are actually fed into the test. For this reason, we
implemented the generation of a uniform real value in $[0\..1)$ by dividing the
output of the generator by $2^{64}$, but we implemented the generation of
uniform 32-bit integer values by \emph{returning first the lower and then the upper
32 bits of each 64-bit generated value}.\footnote{If a real value is generated
when the upper 32 bits of the last value are available, they are simply
discarded.} A possible downside of this approach is 
that we might fail to detect some failure in the high bits (of the $64$-bit,
full output) due to the interleaving process: however, the fact that in our
tests \xorshift generators generate many more failures than those reported previously~\cite{PaLXRNG}
suggests that the approach is well founded.

An important consequence of this choice is that some of the bits are actually
not used at all. When analyzing pseudorandom real numbers in the unit interval,
there is an unavoidable bias towards high bits, as they are more significant.
The very lowest bits have lesser importance and will in any case be perturbed by
numerical errors. For this reason, it is a good practice to run 
tests both on a generator and on its reverse\footnote{That is, on
the generator obtained by reversing the order of the 64 bits returned.}~\cite{PTVNR}.
In our case, this is even more necessary, as the lowest eleven bits returned by the generator \emph{are
not used at all} due to the fact that the mantissa of a 64-bit floating-point
number is formed by 53 bits only.  

A recent example shows the importance of testing the reverse generator. 
\citeN{SaMXA} propose a different way to eliminate linear artifacts: instead of multiplying the
output of an underlying \xorshift generator (with $128$ bits of state and based on $32$-bit shifts) by a
constant, they add it (in $\Z/2^{32}\Z$) with the previous output. Since the sum
in $\Z/2^{32}\Z$ is not linear over $\Z/2\Z$, the result should be
free of linear artifacts.
However, while their generator passes BigCrush, its \emph{reverse} fails
systematically the LinearComp, MatrixRank, MaxOft and Permutation test of
BigCrush, which highlights a significant weakness in its lower bits.

We remark that in this paper we do not pursue the search for
\emph{equidistribution}---the property that all tuples of consecutive values,
seen as vectors in the unit cube, are evenly distributed, as done, for instance,
by~\citeN{PaLXRNG}.
\citeN{BreMEHDS} has already argued in detail that for long-period generators
equidistribution is not particularly desirable, as it is a property of the whole
sequence produced by the generator, and in the case of a long-period generator
only a minuscule fraction of the sequence can be actually used.
Moreover, equidistribution is currently impossible to evaluate exactly for
long-period non-linear generators, and in the formulation commonly used in the
literature it is known to be biased towards the high bits~\cite{LEPFRNG}: for
instance, the \wella generator has been designed to be \emph{maximally equidistributed}~\cite{PLMILPGBLRM2}, and indeed it has measure
of equidistribution $\Delta_1=0$, but the generator obtained by reversing its
bits has $\Delta_1=366$: a quite counterintuitive result, as in general we
expect all bits to be equally important.

Another problem with equidistribution is that it is intrinsically unstable,
unless we restrict its usage to the class of linear generators,
only.
Indeed, if we take a maximally equidistributed sequence, no matter how long,
and we flip the most significant bit of a single element of the sequence, the new sequence
will have the \emph{worst possible} $\Delta_1$.
For instance, by flipping the most significant bit of a single chosen value out
of the output of \wella we can turn its equidistribution measure to
$\Delta_1=4143$. But for any statistical or practical purpose the two sequences
are indistinguishable---we are modifying one bit out of $2^5(2^{1024}-1)$.
However, in general this paradoxical behaviour is not a big issue, because the
modified sequence can no longer be emitted by a linear generator.


We note that since multiplication by an invertible constant induces a
permutation of the space of $64$-bit values (and thus of $t$-tuples of such
values), it preserves some of the equidistribution properties of
the underlying generator (this is true of any bijective scrambling function);
more details will be given in the rest of the paper.

\section{Results for \xorshift[64] generators}
\label{sec:resxorshift}

First of all, \emph{all} generators fail at all seeds the
MatrixRank test from TestU01's SmallCrush suite.\footnote{\citeN{PaLXRNG} reports that \emph{half} of
the generators fail this test, but the authors have chosen to use only 32 of
the 64 generated bits as output bits, in practice applying a kind of \emph{decimation}
to the output of the generator.}
A score-rank plot\footnote{Score-rank plots
are the numerosity-based discrete analogous of the complementary cumulative
distribution function of scores.
They give a much clearer picture than frequency dot plots when the data points are scattered and highly variable.} of the SmallCrush scores for all generators is shown in
Figure~\ref{fig:xorshiftsmallscores}. The plot associates with abscissa $x$ the
number of generators with $x$ or more failures.
We observe immediately that there is a wide range of quality among the
generators examined.
The ``bumps'' in the plot corresponds to new tests failed systematically.

\begin{figure}
\centering
\includegraphics[scale=.6]{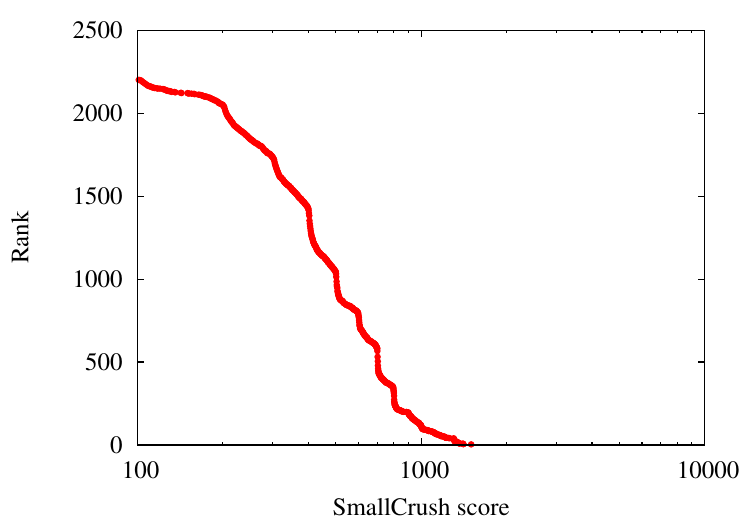}
\caption{\label{fig:xorshiftsmallscores}Score-rank plot of the distribution
of SmallCrush scores for the 2200 possible full-period \xorshift[64] generators.}
\end{figure}


A closer inspection would confirm that there is just a weak correlation
between scores of algorithms conjugate by reversal, because of the 
bias of TestU01 towards high bits. We thus
report in Table~\ref{tab:xorshift} reports the best four generators by combined
scores (i.e., adding the scores of conjugate generators), which are the only
ones failing systematically just the MatrixRank test.
The table reports also results for the generator $A_0(13,7,17)$ suggested by Marsaglia
in his original paper, claiming that it ``will provide an
excellent period $2^{64}-1$ RNG, [\ldots] but any of the above 2200 choices is
likely to do as well''. Clearly, this is not the case:
$A_0(13,7,17)$/$A_1(13,7,17)$ ranks 655 in the combined SmallCrush ranking and
fails systematically several tests. 



 
\begin{sanity} 
Is the result of our experiments dependent on our seed choice? To answer this
question, we repeated our experiments on \xorshift[64] generators with SmallCrush on a 
different set of seeds, namely the integers in the interval $[1\..100]$. 
Kendall's
$\tau$~\cite{KenNMRC,KenTTRP} between the two rankings is $0.98$, which makes it
clear that the dependence on the seed is negligible. In particular, the four best conjugate
pairs in Table~\ref{tab:xorshift} are the same with both seeds.
\end{sanity}

To gather more information, we ran the full BigCrush suite and
Dieharder on our four best generators, on Marsaglia's choice and on the best
choice from ``Numerical Recipes'': the results are given in
Tables~\ref{tab:xorshiftBigCrush} and~\ref{tab:xorshiftDieharder}. Even the four
best generators fail now systematically the BirthdaySpacings, MatrixRank and
LinearComp tests. The
first two generators, however, turn out to perform slightly better than other
two.
We also notice that BigCrush draws a much thicker line between our four best generators and the other ones, which now fail several more tests. Not surprisingly, Dieharder cannot really separate
our four best generators from $A_2(4,35,21)/A_3(4,35,21)$.


\begin{table}\tbl{\label{tab:xorshift}Best four \xorshift[64] generators following SmallCrush.}{%
\renewcommand{\arraystretch}{1.3}
\begin{tabular}{lr|lr|r|r}
Algorithm & Failures & Conjugate & Failures & Overall & $W$\\
\hline
$A_2(11,31,18)$ &     111  &   $A_3(11,31,18)$   &      120 &    231& 25\\
$A_2(8,29,19)$  &     155  &   $A_3(8,29,19)$    &   115   &  270&35\\
$A_0(8,29,19)$  &     159  &   $A_1(8,29,19)$   &   112   &  271&35\\
$A_0(11,31,18)$ &     130  &   $A_1(11,31,18)$  &    150  &   280&25\\
\hline
$A_0(13,7,17)$  &     276  &   $A_1(13,7,17)$   &    802  & 1078& 25  
\end{tabular}}
\end{table} 

\begin{table}\tbl{\label{tab:xorshiftBigCrush}The generators of Table~\protect\ref{tab:xorshift} tested with
BigCrush.}{%
\renewcommand{\arraystretch}{1.3}
\begin{tabular}{lr|lr|r}
Algorithm & Failures & Conjugate & Failures & Overall \\
\hline
$A_2(11,31,18)$ &     762  &   $A_3(11,31,18)$   &  750  &    1512 \\ 
$A_2(8,29,19)$  &  747  &   $A_3(8,29,19)$    &   780   &  1527\\ 
$A_0(8,29,19)$  &     749  &   $A_1(8,29,19)$   &   884   & 1633 \\
$A_0(11,31,18)$ &   748    &   $A_1(11,31,18)$  &  926   & 1674 \\
\hline
$A_2(4,35,21)$  &  961 	    &   $A_3(4,35,21)$   & 1444 & 2405 \\
\hline
$A_0(13,7,17)$  &  1049     &   $A_1(13,7,17)$   &   5454  & 6503 \\
\end{tabular}}
\end{table}

\begin{table}\tbl{\label{tab:xorshiftDieharder}The generators of Table~\protect\ref{tab:xorshift} tested 
with Dieharder.}{%
\renewcommand{\arraystretch}{1.3}
\begin{tabular}{lr|lr|r}
Algorithm & Failures & Conjugate & Failures & Overall\\
\hline
$A_2(11,31,18)$ & 182 & $A_3(11,31,18)$ & 162 & 344\\ 
$A_2(8,29,19)$  & 179 & $A_3(8,29,19)$  & 181 & 360\\
$A_0(8,29,19)$  & 176 & $A_1(8,29,19)$  & 182 & 358\\
$A_0(11,31,18)$ & 181 & $A_1(11,31,18)$ & 186 & 367\\
\hline
$A_2(4,35,21)$  & 189 & $A_3(4,35,21)$  & 187 & 376 \\
\hline
$A_0(13,7,17)$  & 183 & $A_1(13,7,17)$  & 1352 & 1535 \\
\end{tabular}}
\end{table} 

\subsection{Equidistribution}

It is interesting to compare the ranking provided
by equidistribution properties and that provided by statistical tests. 
Note that a \xorshift[64] generator is
$1$-dimensionally equidistributed, that is, every $64$-bit value appears exactly
once except for zero.
We refer to
the already quoted paper by~\citeN{PaLXRNG} for a detailed
description of the equidistribution statistics $\Delta_1$, the \emph{sum of
dimension gaps}: a lower value is better. A \emph{maximally distributed}
generator has $\Delta_1=0$, and we will refer to $\Delta_1$ as to the
\emph{equidistribution score}.
We computed the equidistribution score for all generators using the implementation of Harase's algorithm~\cite{HarELRM}
contained in the \texttt{MTToolBox} package from~\citeN{SaiM}.
Similarly to SmallCrush scores, $\Delta_1$ has 
high-bits bias, and a quite strong one~\cite{LEPFRNG}. 
For a fair comparison, we to thus combine the $\Delta_1$ score of a generator
and of its reverse.  



Figure~\ref{fig:xorshiftsmallequidistcomb} shows that there is some
correlation ($\tau=0.58$) between combined SmallCrush scores and combined equidistribution
scores. Nonetheless, even if equidistribution is able to detect reliably
generators with a very bad SmallCrush score, is not so good at detecting the
generators with the best score, as is visible
from the quite noisy lower left part of the plot.
Indeed, when we restrict our attention
to the best $30$ generators (by combined SmallCrush scores) Kendall's $\tau$ drops to $0.3$.
The first two
generators by combined equidistribution score, $A_4(8,29,19)$ and
$A_6(8,29,19)$, rank $20$ (combined score $361$) and $170$ (score $596$) in the
combined SmallCrush test. When analyzed with the more powerful lens of BigCrush,
they have combined scores $3441$ and $4082$, respectively, and fail
systematically almost \emph{twenty} additional tests with respect to the top 
four generators of Table~\ref{tab:xorshiftBigCrush}. Definitely, choosing
among \xorshift[64] generators by equidistribution score alone is not a good idea.

\begin{figure}
\centering
\includegraphics[scale=.6]{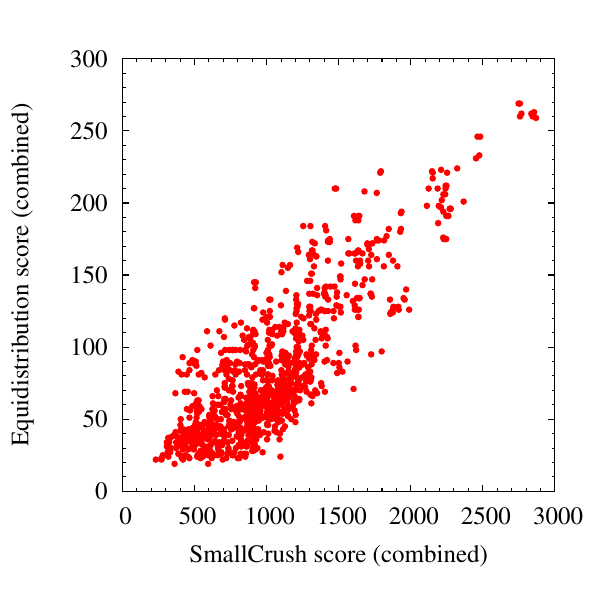}
\caption{\label{fig:xorshiftsmallequidistcomb}Scatter plot of the combined SmallCrush score of conjugate \xorshift[64] generators versus the combined equidistribution score.}
\end{figure}

%

\section{An introduction to \xorshifts[64] generators}

Since a \xorshift[64] generator exhibits evident linearity artifacts, the next
obvious step is to perturb its output using a nonlinear (in $\Z/2\Z$ sense) transformation. A natural candidate is multiplication by a
constant, also because such operation is very fast in modern
processors. Note that the current state of the generator is multiplied by a
constant before returning it, but the state itself is not affected by the
multiplication: thus, the period is the same.

We call such a generator \xorshifts. By choosing a constant invertible modulo
$2^{64}$ (i.e., odd), we can guarantee that the generator will output a
permutation of the sequence output by the underlying \xorshift generator.

This approach was noted in passing in Marsaglia's paper, and it is also proposed
in a more systematic way in the third edition of ``Numerical
Recipes''~\cite{PTVNR} to create a very fast, good-quality pseudorandom number
generator. However, in the latter case the authors \emph{first} compute
allegedly good triples for \xorshift using Diehard (with results markedly
different from ours, and in strident contrast with TestU01's results, as
discussed in Section~\ref{sec:resxorshift}) and \emph{then} choose a multiplier.
There is no reason why the best triples for a
\xorshift[64] generator (which are computed empirically) should continue to be
such in a \xorshifts[64] generator: and indeed, we will see that this is not the
case.

\begin{table}\tbl{\label{tab:mult}The three multipliers used in the rest of the paper. The subscripts recalls the $t$ for which they have good figures of merit.}{%
\centering
$M_{32}=2685821657736338717$ \quad $M_{8}=1181783497276652981$ \quad $M_{2}=
8372773778140471301$}
\end{table}

We thus repeated the experiments of the previous section on \xorshifts[64]
generators. To choose scrambling constants, we followed the heuristic
considerations of~\cite{PTVNR}. We consider primitive (e.g., full-period)
elements of the multiplicative group of $\Z/2^{64}\Z$: these elements have no
fixed point except for zero, which is a very desirable property for a scrambling function.
Moreover, we choose from~\citeN{LEcTLCGDSGLS}
primitive elements that have good qualities as \emph{multiplicative
congruential linear generators}, as we expect that multiplication by such
elements will combine bits in a non-trivial way. 

We use a standard theoretical measure of quality, the \emph{figure of
merit}, which is a normalized best distance between the hyperplanes of families
covering tuples of length $t$ given by successive outputs of the generators 
(see~\citeN{LEcTLCGDSGLS} for details). Since $t$ is an additional parameter,
to further understand the dependency on the multiplier we used three different
multipliers, shown in Table~\ref{tab:mult}, which have good figures of
merit for different $t$'s. The first multiplier, $M_{32}$ (the
one used in~\cite{PTVNR}) and the second, $M_8$, have been taken
from~\citeN{LEcTLCGDSGLS}. The third, $M_2$, was kindly provided by Richard
Simard.

We remark that many other choices for scrambling the output of a generator are
possible, like adding or xoring a fixed word, xoring the output with the output
of another generator, or using a bijective function with strong avalanching
behavior, such as those used in the construction of high-quality hash
functions.
The three factors we considered in our choice are: speed,
good results in statistical test suites, and preservation of 
some equidistribution properties (similarly to the approach taken
in~\cite{LEGCGCDF}).
For instance, xoring with an additive \emph{Weyl generator} (another suggestion
in Marsaglia's paper) makes it in general impossible to prove any
equidistribution property---not even that all $64$-bit value except for zero are output by the
generator. Multiplication by a constant is a very fast operation in modern
processor, and mixing linear operations on $\Z/2\Z$ with operations in
the ring $\Z/2^{64}\Z$ is a standard technique to avoid visible artifacts from
either type of algebraic structure. A drawback is that the lowest bit is, in
fact, not scrambled, and thus it is identical to the lowest bit of
the underlying \xorshift generator.\footnote{As remarked by one of the referees,
since our multipliers are all equal to 1 modulo 4, this is true also of the
second-lowest bit.}

\section{Results for \xorshifts[64] generators}
\label{sec:resxorshifts}

The scatter plot in Figure~\ref{fig:xorshiftlcgcombscatter} shows that there is
essentially no correlation between the scores assigned by SmallCrush to a
generator and its reverse ($\tau=0.15$).\footnote{We report
plots only for $M_{32}$, as the ones for the other multipliers are visually identical.}
Another interesting observation on Figure~\ref{fig:xorshiftlcgcombscatter} is
that the lower right half is essentially empty. So bad generators have a bad
reverse, but there are good generators with a very bad reverse. This suggests
that the quality of a \xorshifts[64] generator can vary wildly from the low to the high bits.

\begin{figure}
\centering
\includegraphics[scale=.6]{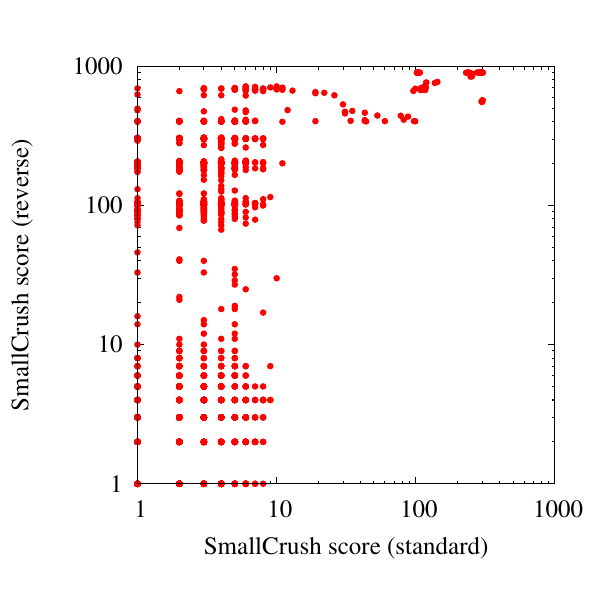}
\caption{\label{fig:xorshiftlcgcombscatter}Scatter plot of the SmallCrush score
of \xorshifts[64] generators and their reverse.}
\end{figure}

A score-rank plot of the SmallCrush scores for all generators shown in
Figure~\ref{fig:xorshiftlcgsmallscores} provides us with further interesting
information: almost all generators have no systematic failure, but only about
half of the reverse generators have no systematic failure. Moreover, the
distribution of standard generators degrades smoothly, whereas the distribution
of reverse generators sports again the ``bump'' phenomenon we observed in
Figure~\ref{fig:xorshiftsmallscores}.

\begin{figure}
\centering
\includegraphics[scale=.6]{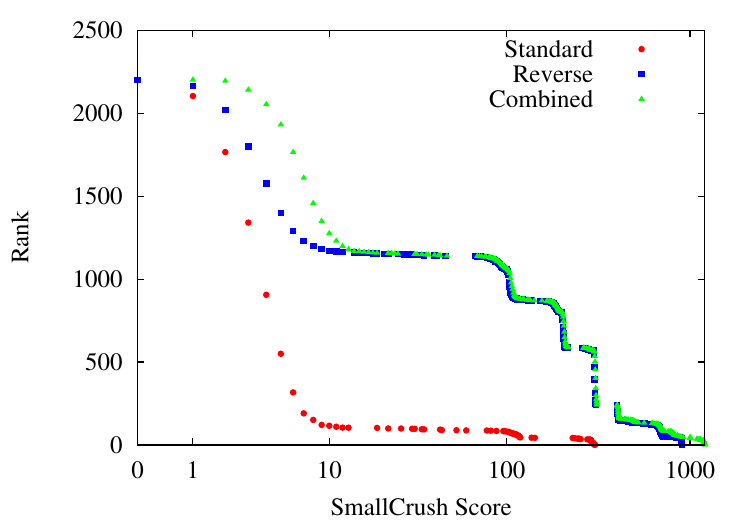}
\caption{\label{fig:xorshiftlcgsmallscores}Score-rank plot of the distribution of SmallCrush scores for the 2200
possible \xorshifts[64] generators with multiplier $M_{32}$.}
\end{figure}

%

Since we need to reduce the number of candidates to apply stronger tests, in the case of $M_{32}$ we
decided to restrict our choice to generators with 3 overall failed tests or
less, which left us with 152 generators. Similar cutoff points were chosen for
$M_8$ and $M_2$.

These generators were few enough so that we could apply both Crush and
Dieharder. Once again, we examine the correlation between the score of a generator
and its reverse by means of the scatter plots in Figure~\ref{fig:xorshiftlcgdh01scatter},
which confirm the high-bits bias, albeit less so in the Dieharder case.

In Figure~\ref{fig:xorshiftlcgdh01} we compare instead the two scores (Crush and
Dieharder) available. The most remarkable feature is there are no points in the upper left corner:
there is no generator that is considered good by Crush but not by Dieharder. On
the contrary, Crush heavily penalizes (in particular because of the score on the
reverse generator) a large number of generators. The generators we will select
in the end all belong to the small cloud in the lower left corner, where the two
test suite agree.

\begin{figure}
\centering
\includegraphics[scale=.6]{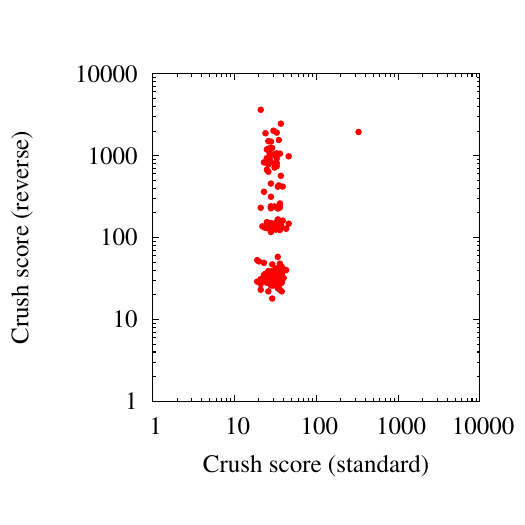}\quad\includegraphics[scale=.6]{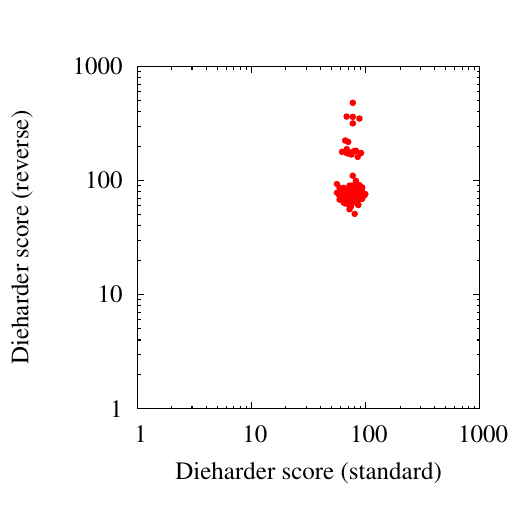}
\caption{\label{fig:xorshiftlcgdh01scatter}Scatter plots for Crush (left)
and Dieharder (right) scores on \xorshifts[64] generators with multiplier $M_{32}$ and their reverse, for the 152 best
generators.}
\end{figure}

\begin{figure}
\centering
\includegraphics[scale=.6]{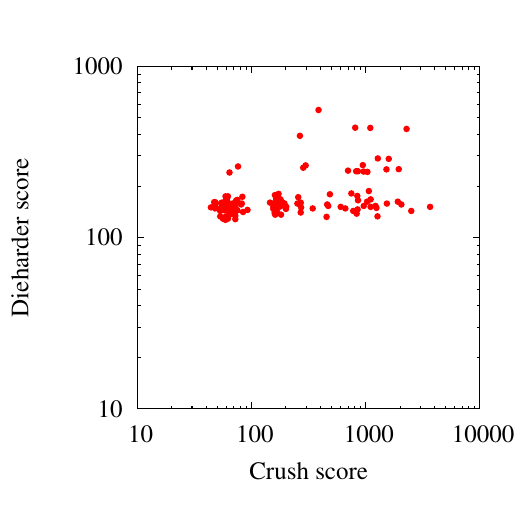}
\caption{\label{fig:xorshiftlcgdh01}A scatter plot of Crush and Diehard
combined scores of the 152 SmallCrush-best \xorshifts[64] generators. The plot is in log-log scale
to accommodate some very high values returned by Crush on reverse
generators. The lower-left ``sweet spot'' corner contains generators
that never fail systematically (not even reversed) in both test suites.}
\end{figure}

\begin{figure}
\centering
\includegraphics[scale=.6]{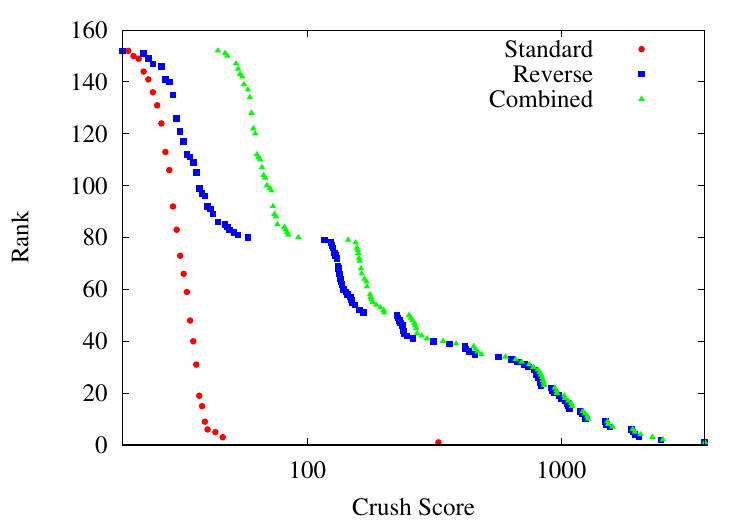}
\caption{\label{fig:xorshiftlcgmediumscores}
Score-rank plot of the distribution of Crush scores for the 
152 SmallCrush-best \xorshifts[64] generators using multiplier $M_{32}$.}
\end{figure}


The score-rank plot in Figure~\ref{fig:xorshiftlcgmediumscores} shows that our strategy
pays off: we started with 152 generators with less than three failures, but
analyzing them with the more powerful lens provided by Crush we get a much
more fine-grained analysis: in particular, only 73 of them give no systematic
failure, and they all belong to the ``sweet spot'' of
Figure~\ref{fig:xorshiftlcgdh01}, that is, they do not give any systematic
failure in Dieharder, too.

Finally, we selected for each multiplier the eight generators with the best
Crush scores, and applied the BigCrush suite: we obtained several generators
failing systematically the MatrixRank test only and shown in Table~\ref{tab:BigCrush64} (which
should be compared with Table~\ref{tab:xorshiftBigCrush}).

\begin{table}\tbl{\label{tab:BigCrush64}Results of BigCrush on the best eight \xorshifts[64] generators found by SmallCrush and Crush in sequence. The generators fail systematically
only MatrixRank.}{%
\renewcommand{\arraystretch}{1.3}
\begin{tabular}{l|rr|r|r}
\multirow{2}{*}{Algorithm} & \multicolumn{3}{c|}{Failures}  & \multirow{2}{*}{$W$}\\
& \multicolumn{1}{c}{S} & \multicolumn{1}{c|}{R} & \multicolumn{1}{c|}{+} \\
\hline
\multicolumn{5}{c}{$M_{32}$}\\
\hline
\xs{7}{11}{5}{45} & 226 & 128 & 354 & 23\\
\xs{7}{17}{23}{52} & 232 & 130 & 362 & 25\\
\xs{1}{12}{25}{27} & 230 & 133 & 363 & 31\\
\xs{1}{17}{23}{29} & 229 & 137 & 366 & 21\\
\xs{5}{14}{23}{33} & 238 & 132 & 370 & 32\\
\xs{5}{17}{47}{29} & 231 & 141 & 372 & 24\\
\xs{1}{16}{25}{43} & 238 & 138 & 376 & 31\\
\xs{7}{23}{9}{57} & 242 & 134 & 376 & 19\\
\hline
\multicolumn{5}{c}{$M_8$}\\
\hline
\xs{5}{11}{5}{32} & 229 & 122 & 351 & 13\\
\xs{2}{8}{31}{17} & 229 & 126 & 355 & 21\\
\xs{5}{3}{21}{31} & 230 & 141 & 371 & 33\\
\xs{3}{17}{45}{22} & 241 & 133 & 374 & 27\\
\xs{4}{8}{37}{21} & 239 & 136 & 375 & 33\\
\xs{3}{13}{47}{23} & 232 & 144 & 376 & 27\\
\xs{3}{13}{35}{30} & 244 & 136 & 380 & 27\\
\xs{4}{9}{37}{31} & 243 & 141 & 384 & 27\\
\hline
\multicolumn{5}{c}{$M_2$}\\
\hline
\xs{7}{13}{19}{28} & 228 & 128 & 356 & 23\\
\xs{3}{9}{21}{40} & 228 & 132 & 360 & 35\\
\xs{1}{14}{23}{33} & 234 & 142 & 376 & 29\\
\xs{7}{19}{43}{27} & 239 & 137 & 376 & 23\\
\xs{1}{17}{47}{28} & 240 & 137 & 377 & 25\\
\xs{5}{16}{11}{27} & 234 & 144 & 378 & 25\\
\xs{4}{4}{35}{15} & 230 & 149 & 379 & 35\\
\xs{7}{13}{21}{18} & 238 & 144 & 382 & 31\\
\end{tabular}}
\end{table}

\subsection{Equidistribution}
\label{sec:equidistlcg}

Multiplication by an invertible element just permutes the elements of
$\Z/2^{64}\Z$ leaving zero fixed, so a \xorshifts[64] generator, like the
underlying \xorshift[64] generator, is $1$-dimensionally
equidistributed.


\section{High dimension}
\label{sec:high}

\citeN{MarXR} describes a strategy for \xorshift generators in high dimension: the idea
is to use always three low-dimensional shifts, but locating them in the context
of a larger $t\times t$ block matrix of the form
\[
M=\left(\begin{matrix}
0 & 0 & 0 & \cdots & 0 & ( I + L^a ) ( I + R^b )\\
I & 0 & 0 & \cdots & 0 & 0\\
0 & I & 0 & \cdots & 0 & 0\\
0 & 0 & I & \cdots & 0 & 0\\
\cdots&\cdots&\cdots&\cdots&\cdots&\cdots\\
0 & 0 & 0 & \cdots & I & ( I + R^c ) \\
\end{matrix}\right)
\]
Marsaglia notes that even in this restricted form there are matrices of full period (he provides examples
for 32-bit shifts up to 160 bits). However, 
this route has not been explored for high-dimensional (say, more than 1024 bits
of state) generators. The only similar approach is that proposed by~\citeN{BreSLPRNGUSX} with his \xorgens generators, which however uses
four shifts. The obvious question is thus: is the additional shift
really necessary to pass a strong statistical test such as BigCrush? We are thus
going to look for good, full-period generators with $1024$ or $4096$ bits of state using 64-bit basic shifts.\footnote{The reason why the number 4096 is
relevant here is that we know the factorization of Fermat's numbers $2^{2^k}+1$
only up to $k=11$. When more Fermat numbers will be factorized, it will be
possible to design \xorshift or \xorgens generators with larger state
space~\cite{BreSLPRNGUSX}. Note that, however, in practice a period of
$2^{1024}-1$ is more than sufficient for any purpose. For example, even if
$2^{100}$ computers were to generate sequences of $2^{100}$ numbers starting
from random seeds using a generator with period $2^{1024}$, the chances that two
sequences overlap would be less than $2^{-724}$.} 

The output of such generators will be given by the last $64$ bits of
the state. It is well known~\cite{BreNMXRNG,NieRNGQMCM} that every bit of
state satisfies a linear recurrence (defined by the characteristic
polynomial) with full period, so \textit{a fortiori} the last $64$ bits
have full period, too. 

Since we already know that
some deficiencies of low-dimensional \xorshift generators are well corrected by
multiplication by a constant, we will follow the same approach, thus looking for
good \xorshifts generators of high dimension.\footnote{As in the \xorshift[64] case, different choices for the shifts
are possible. We will not pursue them here.} Note that since
multiplication by an integer invertible in $\Z/2^{64}\Z$ is a permutation of
$\Z/2^{64}\Z$, a high-dimension \xorshifts generator has the same
period of the underlying \xorshift generator.

We cannot in principle claim full period if we look at a \emph{single}
bit of the output of a \xorshifts generator; but this property can be easily
proved by purely combinatorial means:

\begin{proposition}
\label{prop:bit}
Let \lstt{\bm x}{2^n} be a list of $2^t$-bit values, $t\leq n$, such that every
value appears $2^{n-t}$ times, except for $0$, which appears $2^{n-t}-1$ times.
Then, for every fixed bit $k$ the associated sequence has
period $2^n-1$.
\end{proposition}
\begin{proof}
Suppose that there is a $k$ and a $p\mid 2^n-1$ such that the $k$-th
bit of \lstt{\bm x}{2^n} has period $p$ (that is, the sequence of bits
associated with the $k$-th bit is made by $(2^n-1)/p$ repetitions of
the same sequence of $p$ bits). The $k$-th bit runs through $2^{n-1}-1$ zeroes
and $2^{n-1}$ ones (as there is a missing zero in the output sequence).
This means that $(2^n-1)/p\mid 2^{n-1}$, too, as the same number of ones must appear in every repeating subsequence, and
since $(2^n-1)/p$ is odd this implies $p=2^n-1$.
\end{proof}

\begin{corollary}
Every bit of the output of a full-period \xorshifts generator
has full period.
\end{corollary}

\subsection{Finding good shifts}

The first step is identifying values of $a$, $b$ and $c$ for which the
generator has maximum period using the primitivity check on the characteristic polynomial.
We performed these computations using the algebra package Fermat~\cite{LewF},
with the restriction that $a+b\leq 64$ and that $a$ is coprime with $b$
(see~\cite{BreSLPRNGUSX} for the rationale behind this choices, which
significantly reduce the search space). The resulting sets of values are those
shown in Table~\ref{tab:BigCrush1024star} and~\ref{tab:BigCrush4096}.

For a state of 1024 bits, we obtain 20 possible parameter choices, which we examined in combination with our
three multipliers both through BigCrush and through Dieharder. The results, reported in 
Table~\ref{tab:BigCrush1024star} and~\ref{tab:Dieharder1024}, are excellent: 
with the exception of two pathological choices, no test is failed 
systematically. For a state of 4096 bits (Table~\ref{tab:BigCrush4096}
and~\ref{tab:Dieharder4096}) there are 10 possible parameter choices, 
and no generator fails a test systematically.


\subsection{Equidistribution}

Looking at the shape of the matrix defining high-dimensional \xorshift
generators it is clear that if the state is made of $n$ bits the last
$n/64$ output values, concatenated, are equal to the current state. This implies
that such generators are $n/64$-dimensionally equidistributed (i.e.,
every $n/64$-tuple of consecutive $64$-bit values appears exactly once, except
for a missing tuple of zeroes), so \xorshift[1024] generators are $16$-dimensionally
equidistributed and \xorshift[4096] generators are $64$-dimensionally
equidistributed. Since multiplication by a constant just permutes the space of
tuples, the same is true of the associated \xorshifts generators.

\begin{sidewaystable}\tbl{\label{tab:BigCrush1024star}Results of BigCrush on the \xorshifts[1024] generators. The last two generators fail systematically 
CouponCollector, Gap, HammingIndep, MatrixRank, SumCollector and WeightDistrib.}{%
\renewcommand{\arraystretch}{1.3}
\begin{tabular}{c|c|c}
$M_{32}$&$M_8$&$M_2$\\
\hline
\begin{tabular}{l|rr|r|r}
\multirow{2}{*}{$a$, $b$, $c$} & \multicolumn{3}{c|}{Failures}  & \multirow{2}{*}{$W$}\\
& \multicolumn{1}{c}{S} & \multicolumn{1}{c|}{R} & \multicolumn{1}{c|}{+} \\
\hline
\xst{27}{13}{46} & 25 & 31 & 56 & 275\\
\xst{31}{33}{37} & 28 & 32 & 60 & 79\\
\xst{22}{7}{48} & 37 & 24 & 61 & 223\\
\xst{7}{16}{55} & 37 & 26 & 63 & 65\\
\xst{9}{14}{41} & 23 & 40 & 63 & 167\\
\xst{41}{7}{29} & 28 & 37 & 65 & 265\\
\xst{1}{13}{7} & 34 & 34 & 68 & 113\\
\xst{10}{11}{61} & 32 & 36 & 68 & 155\\
\xst{9}{5}{60} & 44 & 28 & 72 & 227\\
\xst{16}{23}{30} & 37 & 36 & 73 & 59\\
\xst{3}{26}{35} & 45 & 29 & 74 & 89\\
\xst{25}{8}{15} & 42 & 34 & 76 & 281\\
\xst{31}{11}{30} & 35 & 43 & 78 & 363\\
\xst{40}{11}{31} & 38 & 40 & 78 & 77\\
\xst{31}{10}{27} & 34 & 45 & 79 & 233\\
\xst{2}{11}{61} & 43 & 40 & 83 & 81\\
\xst{15}{16}{19} & 45 & 39 & 84 & 255\\
\xst{10}{9}{63} & 39 & 51 & 90 & 69\\
\xst{51}{1}{46} & 31 & 890 & 921 & 111\\
\xst{47}{1}{41} & 50 & 902 & 952 & 99\\
\end{tabular}&
\begin{tabular}{l|rr|r|r}
\multirow{2}{*}{$a$, $b$, $c$} & \multicolumn{3}{c|}{Failures}  & \multirow{2}{*}{$W$}\\
& \multicolumn{1}{c}{S} & \multicolumn{1}{c|}{R} & \multicolumn{1}{c|}{+} \\
\hline
\xst{1}{13}{7} & 28 & 19 & 47 & 113\\
\xst{3}{26}{35} & 29 & 22 & 51 & 89\\
\xst{40}{11}{31} & 24 & 33 & 57 & 77\\
\xst{15}{16}{19} & 30 & 32 & 62 & 255\\
\xst{22}{7}{48} & 29 & 33 & 62 & 223\\
\xst{9}{14}{41} & 32 & 30 & 62 & 167\\
\xst{41}{7}{29} & 25 & 38 & 63 & 265\\
\xst{31}{11}{30} & 33 & 32 & 65 & 363\\
\xst{2}{11}{61} & 25 & 41 & 66 & 81\\
\xst{10}{11}{61} & 42 & 25 & 67 & 155\\
\xst{7}{16}{55} & 32 & 35 & 67 & 65\\
\xst{16}{23}{30} & 35 & 34 & 69 & 59\\
\xst{25}{8}{15} & 25 & 45 & 70 & 281\\
\xst{27}{13}{46} & 39 & 32 & 71 & 275\\
\xst{31}{10}{27} & 40 & 32 & 72 & 233\\
\xst{9}{5}{60} & 40 & 36 & 76 & 227\\
\xst{31}{33}{37} & 39 & 39 & 78 & 79\\
\xst{10}{9}{63} & 31 & 49 & 80 & 69\\
\xst{51}{1}{46} & 60 & 896 & 956 & 111\\
\xst{47}{1}{41} & 67 & 907 & 974 & 99\\
\end{tabular}&
\begin{tabular}{l|rr|r|r}
\multirow{2}{*}{$a$, $b$, $c$} & \multicolumn{3}{c|}{Failures}  & \multirow{2}{*}{$W$}\\
& \multicolumn{1}{c}{S} & \multicolumn{1}{c|}{R} & \multicolumn{1}{c|}{+} \\
\hline
\xst{3}{26}{35} & 29 & 24 & 53 & 89\\
\xst{27}{13}{46} & 41 & 20 & 61 & 275\\
\xst{25}{8}{15} & 38 & 24 & 62 & 281\\
\xst{31}{10}{27} & 36 & 31 & 67 & 233\\
\xst{9}{5}{60} & 24 & 43 & 67 & 227\\
\xst{1}{13}{7} & 28 & 42 & 70 & 113\\
\xst{15}{16}{19} & 36 & 34 & 70 & 255\\
\xst{2}{11}{61} & 40 & 30 & 70 & 81\\
\xst{41}{7}{29} & 36 & 34 & 70 & 265\\
\xst{9}{14}{41} & 33 & 37 & 70 & 167\\
\xst{22}{7}{48} & 37 & 35 & 72 & 223\\
\xst{31}{11}{30} & 45 & 27 & 72 & 363\\
\xst{7}{16}{55} & 36 & 39 & 75 & 65\\
\xst{31}{33}{37} & 37 & 39 & 76 & 79\\
\xst{10}{11}{61} & 41 & 37 & 78 & 155\\
\xst{16}{23}{30} & 44 & 37 & 81 & 59\\
\xst{40}{11}{31} & 38 & 48 & 86 & 77\\
\xst{10}{9}{63} & 48 & 48 & 96 & 69\\
\xst{51}{1}{46} & 31 & 799 & 830 & 111\\
\xst{47}{1}{41} & 47 & 799 & 846 & 99\\
\end{tabular}
\end{tabular}}
\end{sidewaystable}

\begin{sidewaystable}\tbl{\label{tab:Dieharder1024}Results of Dieharder on \xorshifts[1024] generators. No test is failed systematically.}{%
\renewcommand{\arraystretch}{1.3}
\begin{tabular}{c|c|c}
$M_{32}$&$M_8$&$M_2$\\
\hline
\begin{tabular}{l|rr|r|r}
\multirow{2}{*}{$a$, $b$, $c$} & \multicolumn{3}{c|}{Failures}  & \multirow{2}{*}{$W$}\\
& \multicolumn{1}{c}{S} & \multicolumn{1}{c|}{R} & \multicolumn{1}{c|}{+} \\
\hline
\xst{31}{33}{37} & 57 & 67 & 124 & 79\\
\xst{31}{11}{30} & 65 & 61 & 126 & 363\\
\xst{16}{23}{30} & 74 & 56 & 130 & 59\\
\xst{41}{7}{29} & 71 & 61 & 132 & 265\\
\xst{9}{14}{41} & 74 & 64 & 138 & 167\\
\xst{10}{9}{63} & 74 & 66 & 140 & 69\\
\xst{22}{7}{48} & 66 & 75 & 141 & 223\\
\xst{51}{1}{46} & 78 & 63 & 141 & 111\\
\xst{27}{13}{46} & 63 & 79 & 142 & 275\\
\xst{25}{8}{15} & 80 & 64 & 144 & 281\\
\xst{3}{26}{35} & 81 & 66 & 147 & 89\\
\xst{2}{11}{61} & 79 & 71 & 150 & 81\\
\xst{40}{11}{31} & 74 & 76 & 150 & 77\\
\xst{31}{10}{27} & 82 & 71 & 153 & 233\\
\xst{47}{1}{41} & 74 & 79 & 153 & 99\\
\xst{9}{5}{60} & 81 & 75 & 156 & 227\\
\xst{10}{11}{61} & 75 & 84 & 159 & 155\\
\xst{15}{16}{19} & 72 & 88 & 160 & 255\\
\xst{7}{16}{55} & 94 & 68 & 162 & 65\\
\xst{1}{13}{7} & 87 & 76 & 163 & 113\\
\end{tabular}&
\begin{tabular}{l|rr|r|r}
\multirow{2}{*}{$a$, $b$, $c$} & \multicolumn{3}{c|}{Failures}  & \multirow{2}{*}{$W$}\\
& \multicolumn{1}{c}{S} & \multicolumn{1}{c|}{R} & \multicolumn{1}{c|}{+} \\
\hline
\xst{25}{8}{15} & 67 & 56 & 123 & 281\\
\xst{16}{23}{30} & 77 & 54 & 131 & 59\\
\xst{7}{16}{55} & 66 & 66 & 132 & 65\\
\xst{3}{26}{35} & 60 & 75 & 135 & 89\\
\xst{10}{11}{61} & 63 & 74 & 137 & 155\\
\xst{31}{10}{27} & 74 & 69 & 143 & 233\\
\xst{31}{33}{37} & 86 & 58 & 144 & 79\\
\xst{47}{1}{41} & 82 & 62 & 144 & 99\\
\xst{27}{13}{46} & 78 & 69 & 147 & 275\\
\xst{31}{11}{30} & 85 & 62 & 147 & 363\\
\xst{10}{9}{63} & 65 & 86 & 151 & 69\\
\xst{41}{7}{29} & 84 & 68 & 152 & 265\\
\xst{2}{11}{61} & 88 & 65 & 153 & 81\\
\xst{9}{14}{41} & 77 & 80 & 157 & 167\\
\xst{40}{11}{31} & 82 & 78 & 160 & 77\\
\xst{15}{16}{19} & 85 & 76 & 161 & 255\\
\xst{51}{1}{46} & 92 & 74 & 166 & 111\\
\xst{22}{7}{48} & 90 & 82 & 172 & 223\\
\xst{1}{13}{7} & 79 & 95 & 174 & 113\\
\xst{9}{5}{60} & 97 & 89 & 186 & 227\\
\end{tabular}&
\begin{tabular}{l|rr|r|r}
\multirow{2}{*}{$a$, $b$, $c$} & \multicolumn{3}{c|}{Failures}  & \multirow{2}{*}{$W$}\\
& \multicolumn{1}{c}{S} & \multicolumn{1}{c|}{R} & \multicolumn{1}{c|}{+} \\
\hline
\xst{22}{7}{48} & 56 & 76 & 132 & 223\\
\xst{15}{16}{19} & 66 & 67 & 133 & 255\\
\xst{10}{9}{63} & 70 & 71 & 141 & 69\\
\xst{51}{1}{46} & 65 & 78 & 143 & 111\\
\xst{1}{13}{7} & 80 & 64 & 144 & 113\\
\xst{40}{11}{31} & 80 & 67 & 147 & 77\\
\xst{2}{11}{61} & 85 & 65 & 150 & 81\\
\xst{31}{11}{30} & 75 & 75 & 150 & 363\\
\xst{25}{8}{15} & 74 & 77 & 151 & 281\\
\xst{10}{11}{61} & 79 & 76 & 155 & 155\\
\xst{47}{1}{41} & 70 & 86 & 156 & 99\\
\xst{9}{5}{60} & 70 & 86 & 156 & 227\\
\xst{16}{23}{30} & 81 & 76 & 157 & 59\\
\xst{27}{13}{46} & 78 & 80 & 158 & 275\\
\xst{7}{16}{55} & 92 & 70 & 162 & 65\\
\xst{9}{14}{41} & 87 & 80 & 167 & 167\\
\xst{41}{7}{29} & 87 & 81 & 168 & 265\\
\xst{31}{10}{27} & 82 & 87 & 169 & 233\\
\xst{3}{26}{35} & 92 & 79 & 171 & 89\\
\xst{31}{33}{37} & 98 & 88 & 186 & 79\\
\end{tabular}
\end{tabular}}
\end{sidewaystable}

\begin{sidewaystable}\tbl{\label{tab:BigCrush4096}Results of BigCrush on \xorshifts[4096] generators.}{%
\renewcommand{\arraystretch}{1.3}
\begin{tabular}{c|c|c}
$M_{32}$&$M_8$&$M_2$\\
\hline
\begin{tabular}{l|rr|r|r}
\multirow{2}{*}{Algorithm} & \multicolumn{3}{c|}{Failures}  & \multirow{2}{*}{$W$}\\
& \multicolumn{1}{c}{S} & \multicolumn{1}{c|}{R} & \multicolumn{1}{c|}{+} \\
\hline
\xst{14}{41}{15} & 33 & 27 & 60 & 241\\
\xst{5}{22}{27} & 34 & 30 & 64 & 45\\
\xst{30}{29}{39} & 33 & 32 & 65 & 177\\
\xst{25}{3}{49} & 30 & 38 & 68 & 441\\
\xst{7}{12}{59} & 43 & 25 & 68 & 103\\
\xst{19}{34}{19} & 34 & 36 & 70 & 291\\
\xst{12}{11}{61} & 32 & 39 & 71 & 195\\
\xst{5}{27}{21} & 34 & 41 & 75 & 187\\
\xst{23}{26}{29} & 36 & 42 & 78 & 49\\
\xst{11}{9}{25} & 35 & 44 & 79 & 567\\
\end{tabular}&
\begin{tabular}{l|rr|r|r}
\multirow{2}{*}{Algorithm} & \multicolumn{3}{c|}{Failures}  & \multirow{2}{*}{$W$}\\
& \multicolumn{1}{c}{S} & \multicolumn{1}{c|}{R} & \multicolumn{1}{c|}{+} \\
\hline
\xst{5}{22}{27} & 34 & 35 & 69 & 45\\
\xst{5}{27}{21} & 36 & 35 & 71 & 187\\
\xst{25}{3}{49} & 35 & 37 & 72 & 441\\
\xst{7}{12}{59} & 34 & 39 & 73 & 103\\
\xst{11}{9}{25} & 40 & 34 & 74 & 567\\
\xst{12}{11}{61} & 41 & 33 & 74 & 195\\
\xst{19}{34}{19} & 39 & 35 & 74 & 291\\
\xst{14}{41}{15} & 43 & 34 & 77 & 241\\
\xst{30}{29}{39} & 42 & 37 & 79 & 177\\
\xst{23}{26}{29} & 38 & 43 & 81 & 49\\
\end{tabular}&
\begin{tabular}{l|rr|r|r}
\multirow{2}{*}{Algorithm} & \multicolumn{3}{c|}{Failures}  & \multirow{2}{*}{$W$}\\
& \multicolumn{1}{c}{S} & \multicolumn{1}{c|}{R} & \multicolumn{1}{c|}{+} \\
\hline
\xst{11}{9}{25} & 30 & 33 & 63 & 567\\
\xst{5}{27}{21} & 37 & 27 & 64 & 187\\
\xst{25}{3}{49} & 33 & 34 & 67 & 441\\
\xst{19}{34}{19} & 39 & 36 & 75 & 291\\
\xst{23}{26}{29} & 40 & 35 & 75 & 49\\
\xst{30}{29}{39} & 38 & 37 & 75 & 177\\
\xst{12}{11}{61} & 40 & 37 & 77 & 195\\
\xst{14}{41}{15} & 36 & 42 & 78 & 241\\
\xst{7}{12}{59} & 38 & 44 & 82 & 103\\
\xst{5}{22}{27} & 38 & 50 & 88 & 45\\
\end{tabular}
\end{tabular}}
\end{sidewaystable}

\begin{sidewaystable}\tbl{\label{tab:Dieharder4096}Results of Dieharder on \xorshifts[4096] generators.}{%
\renewcommand{\arraystretch}{1.3}
\begin{tabular}{c|c|c}
$M_{32}$&$M_8$&$M_2$\\
\hline
\begin{tabular}{l|rr|r|r}
\multirow {2}{*}{Algorithm} & \multicolumn{3}{c|}{Failures} & \multirow{2}{*}{$W$} \\
& \multicolumn{1}{c}{S} & \multicolumn{1}{c|}{R} & \multicolumn{1}{c|}{+} \\
\hline
\xst{25}{3}{49} & 70 & 70 & 140 & 441\\
\xst{12}{11}{61} & 58 & 83 & 141 & 195\\
\xst{30}{29}{39} & 67 & 77 & 144 & 177\\
\xst{5}{22}{27} & 62 & 84 & 146 & 45\\
\xst{11}{9}{25} & 73 & 75 & 148 & 567\\
\xst{19}{34}{19} & 85 & 66 & 151 & 291\\
\xst{14}{41}{15} & 83 & 74 & 157 & 241\\
\xst{7}{12}{59} & 73 & 85 & 158 & 103\\
\xst{23}{26}{29} & 73 & 88 & 161 & 49\\
\xst{5}{27}{21} & 98 & 67 & 165 & 187\\
\end{tabular}&
\begin{tabular}{l|rr|r|r}
\multirow{2}{*}{Algorithm} & \multicolumn{3}{c|}{Failures}  & \multirow{2}{*}{$W$}\\
& \multicolumn{1}{c}{S} & \multicolumn{1}{c|}{R} & \multicolumn{1}{c|}{+} \\
\hline
\xst{25}{3}{49} & 67 & 70 & 137 & 441\\
\xst{14}{41}{15} & 72 & 69 & 141 & 241\\
\xst{30}{29}{39} & 70 & 75 & 145 & 177\\
\xst{11}{9}{25} & 73 & 77 & 150 & 567\\
\xst{12}{11}{61} & 75 & 80 & 155 & 195\\
\xst{19}{34}{19} & 89 & 67 & 156 & 291\\
\xst{5}{22}{27} & 93 & 65 & 158 & 45\\
\xst{23}{26}{29} & 72 & 87 & 159 & 49\\
\xst{5}{27}{21} & 75 & 84 & 159 & 187\\
\xst{7}{12}{59} & 90 & 77 & 167 & 103\\
\end{tabular}&
\begin{tabular}{l|rr|r|r}
\multirow{2}{*}{Algorithm} & \multicolumn{3}{c|}{Failures}  & \multirow{2}{*}{$W$}\\
& \multicolumn{1}{c}{S} & \multicolumn{1}{c|}{R} & \multicolumn{1}{c|}{+} \\
\hline
\xst{19}{34}{19} & 75 & 64 & 139 & 291\\
\xst{5}{22}{27} & 67 & 77 & 144 & 45\\
\xst{25}{3}{49} & 77 & 71 & 148 & 441\\
\xst{5}{27}{21} & 77 & 71 & 148 & 187\\
\xst{11}{9}{25} & 81 & 76 & 157 & 567\\
\xst{14}{41}{15} & 79 & 78 & 157 & 241\\
\xst{23}{26}{29} & 74 & 84 & 158 & 49\\
\xst{12}{11}{61} & 74 & 85 & 159 & 195\\
\xst{7}{12}{59} & 84 & 79 & 163 & 103\\
\xst{30}{29}{39} & 78 & 89 & 167 & 177\\
\end{tabular}
\end{tabular}}
\end{sidewaystable}

\section{Jumping ahead}

The simple form of a \xorshift generator makes it trivial to jump ahead quickly
by any number of next-state steps. If $\bm v$ is the current state, we want to
compute $\bm v M^j$ for some $j$. But $M^j$ is always expressible as a
polynomial in $M$ of degree lesser than that of the characteristic polynomial.
To find such a polynomial it suffices to compute $x^j \bmod P(x)$, where $P(x)$
is the characteristic polynomial of $M$. Such a computation can be easily
carried out using standard techniques (quadratures to find $x^{2^k}\bmod P(x)$,
etc.), leaving us with a polynomial $Q(x)$ such that $Q(M)=M^j$. Now, if \[
Q(x)=\sum_{i=0}^n \alpha_ix^i, \] we have \[ \bm v M^j = \bm v Q(M)
=\sum_{i=0}^n \alpha_i\bm vM^i, \] and now $\bm v M^i$ is just the $i$-th state
after the current one. If we known in advance the $\alpha_i$'s, computing $\bm
vM^j$ requires just computing the next state for $n$ times, accumulating by xor
the $i$-th state iff $\alpha_i\neq0$.\footnote{Brent's
\texttt{ranut} generator~\cite{BreURNGS} contains one of the first
applications of this technique.}

In general, one needs to compute the $\alpha_i$'s for each desired $j$, but the
practical usage of this technique is that of providing subsequences that are
guaranteed to be non-overlapping. We can fix a reasonable jump, for example
$2^{512}$ for a \xorshifts[1024] generator, and store the $\alpha_i$'s for such a
jump as a bit mask. Operating the jump is now entirely trivial, as it requires
at most $1024$ state changes. In Figure~\ref{fig:jump1024} we show the jump function
for the generator of Figure~\ref{fig:code1024}. By iterating the jump function,
one can access $2^{512}$ non-overlapping sequences of length $2^{512}$ (except for
the last one, which will be of length $2^{512}-1$).

\section{Comparison}

How do our best \xorshifts generators score with respect to more complex
generators in the literature? We decided to perform a comparison with the
popular Mersenne Twister \mt~\cite{MaNMT},\footnote{More precisely, with its
64-bit version.} with \wella/\wellb, two generators introduced by~\citeN{PLMILPGBLRM2} as an improvement over the Mersenne Twister, and
with \xorgens[4096], a very recent 4096-bit generator introduced by~\citeN{BreSLPRNGUSX}
we mentioned in Section~\ref{sec:high}. All these generators are
non-cryptographic and aim at fast, high-quality generation.
As usual, 100 tests are performed at 100 equispaced points of the state space.

We choose generators from the \xorshifts family that
perform well on both BigCrush and Dieharder, have a good weight score and enough
large parameters (which provide faster state change spreading): more
precisely, the \xorshifts[64] generator $A_1(12,25,27)\cdot M_{32}$ (Figure~\ref{fig:code64}), 
\xorshifts[1024] with parameters $31$, $11$, $30$ and multiplier $M_8$
(Figure~\ref{fig:code1024}), and \xorshifts[4096] with parameters $25$, $3$, $49$
and multiplier $M_2$.

\subsection{Quality}

Table~\ref{tab:full01} compares the BigCrush scores of the generators we discussed. 
The results are quite interesting. A simple 64-bit \xorshifts
generator has less linear artifacts than \mt, \wella or \wellb and, thus, a significantly better
score. High-dimension \xorgens[4096] and \xorshifts generators perform significantly
better, in spite of being extremely simple, and have no systematic failure.
The 64-bit \xorshifts generator suggested by ``Numerical Recipes'' fails
systematically the BirthdaySpacings test, contrarily the one we 
have selected.\footnote{Note that we report the number of \emph{failed tests}
on our 100 seeds. L'Ecuyer and Simard~\cite{LESTU01} report the number of \emph{types of failed tests} (e.g.,
failing two distinct RandomWalk tests counts as one) on a single run, so some
care must be taken when comparing the results we report and those reported by
them.}  
We do not report the results of Dieharder, as at this level of quality the suite is 
unable to make any significant distinction among the generators.



\begin{table}\tbl{\label{tab:full01}A comparison of generators using BigCrush.}{%
\renewcommand{\arraystretch}{1.3}
\begin{tabular}{l|rr|r|r|l}
\multirow{2}{*}{Algorithm} & \multicolumn{3}{c|}{Failures} &
\multirow{2}{*}{$W/n$} & \multirow{2}{*}{Systematic}\\
& \multicolumn{1}{c}{S} & \multicolumn{1}{c|}{R} & \multicolumn{1}{c|}{+}&\\
\hline
$A_1(12,25,27)\cdot M_{32}$  & 230 & 133 & 363 & $0.48$ &MatrixRank \\ 
$A_3(4,35,21)\cdot M_{32}$   & 240 & 223 & 463 & $0.38$ &MatrixRank, BirthdaySpacings \\ 
\xorshifts[1024]  & 33 & 32 & 65 & $0.35$ &---\\  
\xorshifts[4096]  & 33 & 34 & 67 & $0.11$ &--- \\ 
\xorgens[4096]		 & 42 &  40 & 82 & $0.23$ &--- \\ 
\mt          & 258 & 258 & 516 & $0.34$ &  LinearComp \\ 
\wella       & 441 & 441 & 882 & $0.40$ &MatrixRank, LinearComp \\ 
\wellb       & 235 & 233 & 468 & $0.43$  &LinearComp \\ 
\end{tabular}}
\end{table}


\subsection{Escaping zeroland}

We show in Figure~\ref{fig:ez} the speed at which a few of the
generators of Table~\ref{tab:full01} ``escape from zeroland''~\cite{PLMILPGBLRM2}:
purely linearly recurrent generators with a very large state space need a very
long time to get from an initial state with a small number of ones to a state in
which the ones are approximately half. The figure shows a measure of escape time
given by the ratio of ones in a window of 4 consecutive 64-bit values sliding over the first 100\,000 generated values, averaged
over all possible seeds with exactly one bit set (see~\cite{PLMILPGBLRM2} for a
detailed description).

As it is known, \mt needs hundreds of thousands of iterations to start
behaving correctly. \xorshifts[4096] and \xorgens[4096] need a few thousand (but \xorgens[4096] oscillates always around $1/2$), \wellb and \xorshifts[1024] a few hundreds,
whereas \wella just a few dozens, and \xorshifts[64] is almost unaffected. 

Table~\ref{tab:ez} condenses
Figure~\ref{fig:ez} into the mean and standard deviation of the displayed values.
Clearly, the multiplication step helps in 
reducing the correlation between the number of ones in the state and the
number of ones in the output values. Also, the slowness in recovering from states
with too many zeroes it directly correlated to the size of the state space---a very good
argument against linear generators with too large state spaces. 

\begin{figure}
\centering
\includegraphics[scale=.8]{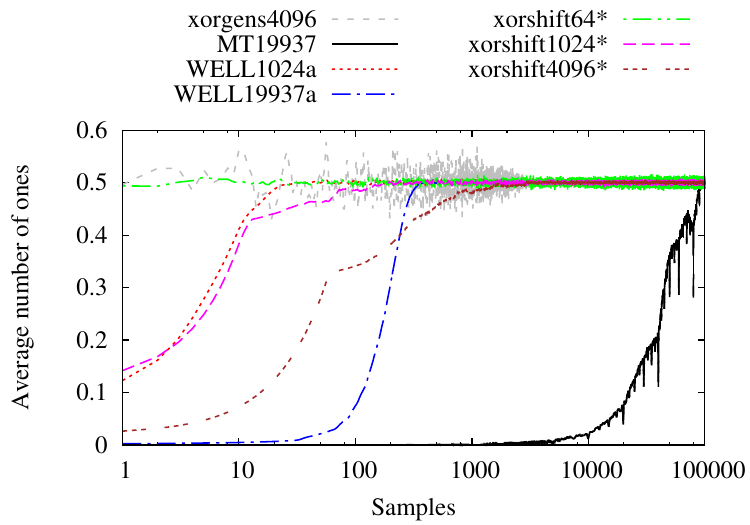}
\caption{\label{fig:ez}Convergence to ``half of the bits are ones in average'' plot.}
\end{figure}

\begin{table}\tbl{\label{tab:ez}Mean and standard deviation for the data shown in Figure~\protect\ref{fig:ez}.}{%
\renewcommand{\arraystretch}{1.3}
\begin{tabular}{l|rr}
Algorithm & Mean & Standard deviation\\
\hline
\xorshifts[64]   & $0.5000$ & $0.0039$ \\
\xorgens[4096]   & $0.5000$ & $0.0031$ \\
\xorshifts[1024] & $0.5000$ & $0.0035$ \\
\wella           & $0.4999$ & $0.0036$ \\
\xorshifts[4096] & $0.4992$ & $0.0110$ \\
\wellb           & $0.4983$ & $0.0185$ \\
\mt              & $0.2823$ & $0.1705$ \\
\end{tabular}}
\end{table}

\subsection{Speed}

Finally, we benchmark the generators of Table~\ref{tab:full01}. Our tests were run
on an Intel\textregistered{} Core\texttrademark{} i7-4770 CPU @3.40GHz
(Haswell), and the results are shown in Table~\ref{tab:speed} (variance is
undetectable, as we generate $10^{10}$ values in each test). We also report as a strong baseline results about \texttt{SFMT19937}, the
\emph{SIMD-Oriented Fast Mersenne Twister}~\cite{SaMSOFMT}, a 128-bit version of the Mersenne Twister
based on the SSE2 extended instruction set of Intel processors (and thus not usable, in principle, on other processors).
We used suitable options to keep the compiler from unrolling loops or extracting
loop invariants.

The highest speed is achieved by the high-dimensional \xorshifts generators. 
\texttt{SFMT19937} is a major improvement in speed over \mt, albeit
slightly slower than a high-dimensional \xorshifts generator; it fails systematically, moreover,
the same tests of \texttt{MT19937}.

A \xorshifts[64] generator is actually \emph{slower} than its
high-dimensional counterparts. This is not surprising, as the three shift/xors in a
\xorshifts[64] generator form a dependency chain and must be executed in
sequence, whereas two of the shifts of a higher-dimension generator are
independent and can be internally parallelized by the CPU.
\wella and \wellb are heavily penalized by their 32-bit structure.

\begin{table}\tbl{\label{tab:speed}Time to emit a 64-bit integer on an
Intel\textregistered{} Core\texttrademark{} i7-4770 CPU @3.40GHz (Haswell).}{%
\renewcommand{\arraystretch}{1.3}
\begin{tabular}{l|rrr}
Algorithm & Speed (ns/64 bits)\\
\hline
\xorshifts[64] & $1.58$\\
\xorshifts[1024] & $1.36$\\
\xorshifts[4096] & $1.36$\\
\xorgens[4096]	 & $2.06$\\
\mt (64-bit version)       & $2.84$\\
\texttt{SFMT19937} & $1.80$ \\
\wella     & $10.31$\\
\wellb     & $7.45$\\
\end{tabular}}
\end{table}

\begin{figure}[ht]
\centering
\begin{verbatim}
#include <stdint.h>

uint64_t x;

uint64_t next(void) {
    x ^= x >> 12; // a
    x ^= x << 25; // b
    x ^= x >> 27; // c
    return x * UINT64_C(2685821657736338717);
}
\end{verbatim}
\caption{\label{fig:code64}The suggested \xorshifts[64] generator in C99 code.
The variable \texttt{x} should be initialized to a nonzero seed before calling
\texttt{next()}.}
\end{figure}

\begin{figure}[ht]
\centering
\begin{verbatim}
#include <stdint.h>

uint64_t s[16];
int p;

uint64_t next(void) { 
    const uint64_t s0 = s[p];
    uint64_t s1 = s[p = (p + 1) & 15];
    s1 ^= s1 << 31; // a
    s[p] = s1 ^ s0 ^ (s1 >> 11) ^ (s0 >> 30); // b,c
    return s[p] * UINT64_C(1181783497276652981);
}
\end{verbatim}
\caption{\label{fig:code1024}The suggested \xorshifts[1024] generator in C99 code.
The array \texttt{s} should be initialized to a nonzero seed before calling
\texttt{next()}.}
\end{figure}

\begin{figure}[ht]
\centering
\begin{verbatim}
#include <stdint.h>
#include <string.h>

void jump(void) {
    static const uint64_t JUMP[] = {
        0x84242f96eca9c41d, 0xa3c65b8776f96855, 0x5b34a39f070b5837,
        0x4489affce4f31a1e, 0x2ffeeb0a48316f40, 0xdc2d9891fe68c022,
        0x3659132bb12fea70, 0xaac17d8efa43cab8, 0xc4cb815590989b13,
        0x5ee975283d71c93b, 0x691548c86c1bd540, 0x7910c41d10a1e6a5,
        0x0b5fc64563b3e2a8, 0x047f7684e9fc949d, 0xb99181f2d8f685ca,
        0x284600e3f30e38c3 
    };

    uint64_t t[16] = { 0 };
    for(int i = 0; i < sizeof JUMP / sizeof *JUMP; i++)
        for(int b = 0; b < 64; b++) {
            if (JUMP[i] & 1ULL << b)
                for(int j = 0; j < 16; j++)
                    t[j] ^= s[(j + p) & 15];
            next();
        }

    for(int j = 0; j < 16; j++)
        s[(j + p) & 15] = t[j];
}
\end{verbatim}
\caption{\label{fig:jump1024}The jump function for the \xorshifts[1024] generator of Figure~\ref{fig:code1024}
in C99 code. It is equivalent to $2^{512}$ calls to \texttt{next()}.}
\end{figure}

%
%

\section{Conclusions}

After our careful experimental analysis, we reach the following conclusions:

\smallskip\noindent\textbf{A \xorshifts[1024] generator is an excellent choice for a
general-purpose, high-speed generator.} The statistical quality of the generator
is very high (it has, actually, the best results in BigCrush), and its period is
so large that the probability of overlapping sequences is practically zero, even
in the largest parallel simulation (and strictly non-overlapping sequences can be easily generated
using the jump function).
Nonetheless, the state space is reasonably small, so that seeding it with
high-quality bits is not too expensive, and recovery from states with a large
number of zeroes happens quickly. The generator is also blazingly fast (it
is actually the fastest generator we tested).
The reasonable state space makes it also easier, in case a large number
of generators is used at the same time, to fit their state into the
cache.
In any case, with respect to other generators, the state is accessed in a more
localized way, as read and write operations happen \emph{at two consecutive locations}, and thus will generate 
at most one cache miss.

\smallskip\noindent\textbf{In case memory is an issue, or array access is
expensive, a very good general-purpose generator is a \xorshifts[64] generator.}
While the generator $A_1(12,25,27)\cdot M_{32}$ fails systematically the
MatrixRank test, it has less linear artifacts than \mt, \wella or \wellb, which
fail systematically even more tests. It is a very good choice if memory footprint is an
issue and a very large number of generators is necessary. It can also be used, for instance, to
generate the initial state of another generator with a larger state space using a $64$-bit seed.
We remark that a \xorshifts[64] generator can also actually be \emph{faster} than a
\xorshifts[1024] generator if the underlying language incurs significant
costs when accessing an array: for instance, in Java a \xorshifts[64] generator
emits a value in $1.62$\,ns, whereas a \xorshifts[1024] generator needs
$2.06$\,ns.

\smallskip\noindent\textbf{Linear generators with an excessively long period have a
number of problems that are not compensated by higher statistical quality.}
\wellb is almost four slower than
\xorshifts[1024], and has a worse performance in BigCrush; moreover, recovery from states
with many zeroes, albeit enormously improved with respect to \mt, is still very slow, and seeding properly
the generator requires almost twenty thousands random bits. In the end,
it is in general difficult to motivate state spaces larger than $2^{1024}$.
Similar considerations are made by \citeN{PTVNR} and \citeN{LEPFRNG}.

\smallskip\noindent\textbf{Surprisingly simple and fast generators can produce
sequences that pass strong statistical tests.} The code in
Figure~\ref{fig:code1024} is extremely shorter and simpler than that of \mt,
\wella or \wellb. Yet, it performs significantly better on BigCrush. It is a tribute
to Marsaglia's cleverness that just eight logical operations, one addition and one multiplication by a constant can 
produce sequences of such high quality. \xorgens generators are similar with this respect,
but use several more operations due to the additional shift and to combination with a \emph{Weyl generator} to 
hide linear artifacts~\cite{BreSLPRNGUSX}.

\smallskip\noindent\textbf{The $t$ for which the multiplier has a good figure of merit has no detectable
effect on the quality of the generator.} If our tests, we could not find any significant difference
between the behavior of generators based on $M_{32}$, $M_8$ or $M_2$. It could be interesting 
to experiment with multipliers having very \emph{bad} figures of merit, or more
generally with multipliers chosen using different heuristics.

\smallskip\noindent\textbf{Equidistribution is more useful as a design feature
than as an evaluation feature.} While \emph{designing} generators around
equidistribution might be a good idea, as it leads in general to good
generators, \emph{evaluation} by equidistribution is a more delicate matter
because of high-bits bias, instability issues, and failure to detect the generators having the
best scores in statistical suites.

\smallskip\noindent\textbf{TestU01 has significantly more resolution than
Dieharder as a test suite.} In particular in the high-dimension case, TestU01 is
able to provide useful information, whereas Dieharder scores flatten down.
However, TestU01 (as any other test suite with high-bits bias) must always be
applied to the reverse generator, too.

\bibliography{biblio}

\end{document}